\documentclass[submission,copyright,creativecommons]{eptcs}
\providecommand{\event}{FROM 2026} % Name of the event you are submitting to

\usepackage{MnSymbol} 
\usepackage{verbatim}
\usepackage{amsmath}
\usepackage{amsthm}
\theoremstyle{definition}
\newtheorem{definition}{Definition}
\newtheorem{theorem}{Theorem}
\newtheorem{lemma}{Lemma}

\usepackage{iftex}

\ifpdf
\usepackage{underscore}         % Only needed if you use pdflatex.
\usepackage[T1]{fontenc}        % Recommended with pdflatex
\else
\usepackage{breakurl}           % Not needed if you use pdflatex only.
\fi

\usepackage{coqdoc}
\def\repository{\url{https://www.risc.jku.at/people/tjebelea/Dramnesc-Jebelean-Stratulat-FROM-2026.html}}
\newcommand{\mma}{{\em Mathematica}}
\newcommand{\tma}{{\em Theorema}}

\newcommand{\fa}[1]{{\underset{#1}{\forall}}}
\newcommand{\cfa}[2]{{\underset{#2}{\underset{#1}{\forall}}}}

\newcommand{\mse}[1]{\{\!\{#1\}\!\}} % multiset with specified elements
\newcommand{\ms}{{\cal M}}

\newcommand{\el}{\langle\rangle} % empty list
\newcommand{\cons}{\smallsmile}
\def\Tail{{\em Tail}}
\newcommand{\append}{\smallfrown}
\newcommand{\concat}{\asymp}
\newcommand{\head}{\textit{head}}
\newcommand{\last}{\textit{last}}

\newcommand{\patsortI}{\textit{PatSort1}}
\newcommand{\patsortII}{\textit{PatSort2}}
\newcommand{\patsortIII}{\textit{PatSort3}}
\newcommand{\selsorted}{\textit{SelSorted}}
\newcommand{\selrest}{\textit{SelRest}}
\newcommand{\patcomb}{\textit{PatComb}}
\newcommand{\patjoin}{\textit{PatJoin}}
\newcommand{\patmix}{\textit{PatMix}}
\newcommand{\patsplit}{\textit{PatSplit}}
\newcommand{\merge}{\textit{Merge}}
\newcommand{\issorted}{\textit{IsSorted}}
	\title{Certification of Bilateral Patience Sort\\ in \tma\ and Rocq}
	
	\def\titlerunning{Certification of Patience Sort}
	
	\author{
		Isabela Dr\u amnesc
		\institute{Dept. of Computational Sciences and Artificial Intelligence} \institute{West University of Timi\c soara, Romania}  \email{Isabela.Dramnesc@e-uvt.ro}
		\and
		Tudor Jebelean
		\institute{RISC, Johannes Kepler University, Linz, Austria}
		\institute{ICAM, West University of Timi\c soara, Romania} %\url{https://www3.risc.jku.at/home/tjebelea}
		\email{Tudor.Jebelean@e-uvt.ro}
		\and
		Sorin Stratulat
    		\institute{Université de Lorraine, CNRS, LORIA}
            \institute{F-57000 Metz, France}
	\email{sorin.stratulat@univ-lorraine.fr}
}

\def\authorrunning{I. Dr{\u a}mnesc, T. Jebelean \& S. Stratulat}
\begin{document}
	%\Copyright{Isabela Dr\u amnesc and Tudor Jebelean and Sorin Stratulat}
	%\ccsdesc[500]{Theory of computation~Program verification}
	%\keywords{formal verification; sorting algorithms; Theorema; Rocq}

	%%%%%%%%%%%%%%%%%%%%%%%%%%%%%%%%%%%%%%%%%%%%%%%%%%%%%%
	
	\maketitle
	
	\begin{abstract}
		This is a case study on a specific version of the Patience Sort algorithm in which we illustrate the evolution of it from an intuitive but inefficient nested recursion into a more complex but more efficient tail recursion, together with its formal implementation  and certification in \tma\ and Rocq (formerly Coq).
		We identify some general principles of algorithm transformation, and we develop the necessary background theory and proof methods.
		As a significant distinctive aspect, the approach in \tma\ uses multisets, which simplifies the whole process and makes it more intuitive.
		The certification process reveals significant differences between the \tma\ and Rocq approaches, about which we provide a comparative analysis with respect to algorithm definition, proof development, and proof effort. This analysis offers insights into how algorithm design influences the complexity and structure of formal proofs and demonstrates how non-trivial algorithms can be effectively verified across different formal frameworks.
		
	\end{abstract}
	
	\section{Introduction} %ISA (2 pages)
	
	Sorting is one of the most fundamental problems in computer science, with applications ranging from data processing and optimization to algorithm design and formal reasoning. In addition to its classical role, sorting also plays an important role in modern domains such as robotics, where ordered data structures are required for tasks like sensor data processing, prioritization of candidate actions, trajectory selection, and real-time decision-making.
    Since efficient and correct ordering mechanisms are essential for ensuring reliable system behavior, ensuring the correctness of sorting algorithms is crucial. Even minor implementation errors may lead to incorrect outputs that propagate through higher-level components. This is particularly critical in safety-sensitive systems, including autonomous and robotic platforms, where incorrect decisions based on improperly processed data can have significantly negative consequences.
	
	In this case study we demonstrate the use of \tma\ and Rocq (formerly Coq) for expressing several versions of the Patience Sort algorithm together with the corresponding formal correctness proofs. These systems enable the development of certified implementations whose correctness is formally guaranteed.

    Moreover we study the systematic transformation of the obvious intuitive but inefficient version of the algorithm into a tail-recursive version.
    Tail recursion improves in general the efficiency of the execution because it does not need a stack for the management of intermediate results and function calls. Moreover such algorithms are transformed by good compilers into iterative algorithms.  
	
	The \tma\ system \cite{Buchberger97asurvey,theorema-2,theorema-2-windsteiger} is a framework built upon \textit{Mathematica}\footnote{\url{www.wolfram.com/mathematica}} that supports the processes of defining mathematical theories, including definition of algorithms by logical formulae, experimenting by running the algorithms, and developing and using mechanical provers by specifying inference rules and strategies.
	The system facilitates the certification of algorithms because their implementation in \tma\ does not use a programming language, rather they can be defined directly in predicate logic together with their correctness conditions.
    For this case study we constructed a specific theory and implemented a prover for lists over an arbitrary ordered domain.
	A distinctive feature of the \tma\ system is the use of natural style (similar to human) to express logical formulae and the algorithms, for the inference rules of the provers, and for the presentation of the proofs.
    This makes both the underlying theories as well as the proofs easy to read and understand by humans.
	
	The Rocq system~\cite{coq-manual,coq} is a state-of-the-art proof assistant that supports both the formalization of algorithms and the verification of their properties within a unified framework. It is based on the Calculus of Inductive Constructions, a higher-order type theory that enables expressive specifications and constructive reasoning.  Following the Curry--Howard correspondence~\cite{Howard80}, propositions are represented as types, while proofs correspond to terms inhabiting these types. A key feature of Rocq is its small trusted kernel, which implements a type-checking algorithm to verify that a proof term indeed has the type corresponding to the property it claims to establish. This mechanism ensures a high degree of reliability, making Rocq a widely adopted tool for the certification of algorithms and formal reasoning.
	
	While a substantial work exists on formally verified sorting algorithms, most efforts focus on well-known algorithms such as Insertion Sort, Merge Sort, or Quick Sort. Comparatively, less attention has been given to structurally richer algorithms such as Patience Sort\footnote{We did not find any reference to a mechanized certification of it in the literature.}, which is very important in the context of longest increasing subsequences and combinatorial properties of permutations~\cite{aldous1999longest}, and whose layered structure and non-standard operational principles make it an interesting subject for formal analysis.
    Moreover, from the efficiency point of view, this algorithm is very interesting as, in contrast to all other ones, its recursive operation does not reduce all sublists to small fixed size ones. 
	
	Furthermore, there is a lack of comparative studies investigating how different formal systems support the certification of non-trivial algorithms, both in terms of expressiveness and proof effort. In particular, the interplay between algorithm design and formal reasoning across distinct frameworks remains underexplored.
	
	In this paper, we address these gaps by studying some novel versions of the Patience Sort algorithm and formally certifying their correctness in two different systems: Theorema and Rocq.
    Patience Sort \cite{Mallows73} decomposes the input in a finite number of sorted sublists by adding one-by-one the elements of the input at the beginning of the sublist where it is smaller than the first element (new sublist if necessary!), and then merges the sublists.
    Our specific version is a simpler binary divide and conquer: only one sorted sublist is produced, but by {\em bilaterally} adding the element either at the beginning or at the end if it is smaller than the first, respectively bigger than the last element\footnote{This approach has the potential to improve the efficiency because it produces larger and fewer sublists.}, and then the remaining sublist is sorted recursively in the same way and merged with the sorted one.
    Our approach highlights both the algorithmic aspects of the proposed method and the methodological differences in formal certification.
	Moreover this constitutes a case study illustrating the principles of algorithm transformation into a more efficient recursion scheme.
	
	The relevant files in \tma\ and Rocq are publicly available at \repository\ (we refer to this throughout the paper as {\em the companion repository}).%{\url{https://www.risc.jku.at/people/tjebelea/Dramnesc-Jebelean-Stratulat-FROM-2026.html}}.
	\subsection*{Related Work} 
	Formal certification environments like Rocq~\cite{coq} and \sloppy{Isabelle/HOL~\cite{isabelle}} have been used to verify classical algorithms on arrays/lists~\cite{FilliatreMagaud99,Tushkanova:2009aa,Petrovic:2014aa,Zhang2018,Burgos:2019aa,Bockenek:2019aa,Quarfot-Orrevall:2020aa,Sternagel-2013}.
	Other approaches use Krakatoa~\cite{Marche:2004aa} in~\cite{Tushkanova:2009aa}, SPARK*~\cite{Hoang:2015aa} in~\cite{Baity:2021aa}, PVS~\cite{Shankar:2020ww} in~\cite{bitonic-sort}, KeY~\cite{key} in \cite{Bove-2006}, and ACL2~\cite{Kaufmann:1999ys} in~\cite{ACL2-in-place-Quicksort}. These approaches differ in their underlying logics, proof mechanisms, and degrees of automation.
	
	Across these works, sorting algorithms are typically formalized in a functional style, with definitions that are largely syntactic variants of common core concepts, although some approaches adopt imperative formulations~\cite{Beckert-2017,Safari-2020,Gouw-2016}, leading to different verification challenges.
	
	A key aspect of correctness proofs is establishing that the output is a permutation of the input. This notion admits multiple formalizations, including element removal~\cite{bb-lazy}, index-based bijections~\cite{Jiang-2017}, element counting~\cite{Certezeanu-2016}, filtering-based comparisons~\cite{Georgiou-2024}, and multiset-based approaches~\cite{Burgos:2019aa}. These choices significantly impact the structure and complexity of the proofs, often making them more involved than those for sortedness. Proofs also differ in their inference rules, strategies, and reliance on auxiliary definitions and lemmas, affecting proof length, automation, and readability.
	We follow a distinct specific approach to this topic, namely the use of {\em multisets}. This facilitates the expressing of the fact that two lists have the same elements, and also eases the extension of ordering and inclusion relations to lists and multisets, which opens the possibility of introducing very efficient inference rules for proving over this domains, in particular for the certification of sorting algorithms.
	
	Natural-style formal reasoning, where definitions and proofs resemble traditional mathematical practice, is addressed in relatively few works, including~\cite{isa-jsc2015,jlamp,tokyo,bb-lazy}.
	
	Despite the extensive literature, there is no prior work in Rocq addressing the verification of the class of algorithms considered here, and comparative studies between systems with different paradigms, such as Theorema and Rocq, are still relatively rare. Our work in this paper fills this gap.

    Concerning the automatic transformation of recursive algorithms into tail-recursive ones, the most relevant papers agree that this not yet achieved.
    In \cite{NISHIDA201453} we find an algebraic method for transforming recursive programs into tail-recursive ones, which is also proven correct, however this method works only for a subclass of programs, and the authors consider that the generation of accumulators is not yet automated.
    The survey \cite{Maghawry-2019}, which focuses on the use of genetic algorithms, and the older survey \cite{Visser2005} present various techniques that can be used for program transformation, but it does not address a comprehensive automation process for the whole class of recursive algorithms.
    Finally \cite{PoloLavalleMolina2023} presents an Agda \cite{agda} implementation of a transformation into tail-recursive programs for a certain simple subclass of recursive algorithms.
    All the concrete examples presented in these papers are very simple (like e.g. reverse of a list) and cannot be really applied to the complex functions that are necessary for sorting.

    Nevertheless the latter work is also interesting as it uses Agda, which is, like \tma\ and Rocq, a system allowing both the definition, execution, and certification of algorithms.
    In particular, Agda is used in an interesting unpublished case study\footnote{Owen Stephens, {\em Agda Patience Sort}, \url{https://www.owenstephens.co.uk/blog/2015/07/06/agda-patience-sort.htm}} for the specification and partial certification of Patience Sort, but only few simple properties are proven.
    
    %\marginpar{TODO Tudor}
	
	\paragraph{The novelty of our paper consists in:}
	\begin{itemize}
		\item We present three new versions of the Patience Sort algorithm, ranging from a straightforward intuitive but inefficient approach to a tail-recursive efficient one.
		\item We illustrate using this case study some principles of algorithm transformation from nested recursion to more efficient tail recursion.
		\item We develop in \tma\ a comprehensive theory of lists and of multisets over ordered domains, as well as a dedicated prover with both general and specific inference rules.
		\item We provide formal correctness certifications of the algorithms in both Theorema and Rocq.
        \item We analyze comparatively the two approaches, focusing on expressiveness, proof structure, and proof effort.
        \item We introduce a novel method for stating the correctness
        of tail-recursive functions, using existential quantification.
	\end{itemize}
	It is not our goal in this paper to estimate the running time\footnote{Both \tma\ and Rocq provide direct execution of programs defined as logical formulae, which we use for immediate testing and for demonstrative purposes. Some concrete running examples are presented in the companion repository.} of the algorithms or to investigate their time and space complexity, because we focus on the principles of algorithm transformation and on the certification techniques, including the development of the appropriate theories and proof methods.
	
	\section{The Patience Sort Algorithm} %Tudor (2 pages)
	
	Our reason for choosing this particular algorithm stems from a natural intuitive approach that fits interestingly in the landscape of the most important sorting algorithms (see also \cite{Darlington78}).
	Namely one can classify these algorithms into {\em selection}-based and {\em split}-based.
	Selection algorithms proceed by selecting one element of the input list (or a small constant number), sorting the rest, and then adding the element to the sorted sublist.
	Split algorithms (also known as {\em divide-and-conquer}) proceed by splitting the input into two sublists, sorting them, and the combining the results into the final sorted list.
	In both cases two approaches may occur:
	easy decomposition
	% (take first element, or last, respectively simply break the list in the middle)
	but clever composition
	% of the sorted sublists (insert the element into the right place, respectively merge the sublists)
	(insert-sort, merge-sort)
	as opposed to clever decomposition
	% (take maximum, or minimum, respectively separate elements smaller/bigger than a pivot)
	which then needs only a simple composition
	% (place the element at the beginning, or at the end, respectively concatenate the sublists)
	(max/min-sort\footnote{In \cite{our-LPAR-2024} we show that bubble-sort can be obtained as a tail-recursive version of min/max-sort.}, quick-sort).
	
	In \cite{our-synthesis} we discover an approach that is different from all of the above:  divide-and-conquer using both clever decomposition and composition, but in such a way that {\em only one } of the sublists needs to be sorted! Decomposition is performed by scanning the input list and adding each element at the beginning (or end) of an accumulating sorted sublist when this is possible, or placing it into the second sublist otherwise.
	Composition is performed by the known {\em Merge}\footnote{Merging two sorted lists into a sorted one is performed by scanning the inputs in parallel and moving at each step the smallest of the two elements at the end of resulting list.} method after sorting recursively the second sublist.
	
	The {\em Patience Sort} algorithm \cite{Mallows73} (attributed to 
	A. S. C. Ross in \cite{Mallows63}) also known as Floyd's game is actually a little more complex than the one we present here.
	Namely, this constructs as many sublists as necessary by moving each scanned input element at the beginning of the first sublist where it is smaller than the first one, or starts a new sublist if there is no suitable sublist.
	Composition is performed by a $k$-ary version of {\em Merge}: at each step the smallest of the heads of all the sublists is moved to the accumulating final sorted list. This can be expressed in fact as a tail-recursive algorithm, which exhibits however some differences from \patsortIII\ that we obtain below.
	
	\subsection{Basic Notions and Notations}

	We consider multisets and lists over a {\bf totally ordered domain} (notation $<$ and $\leq$), whose elements are denoted by $a, b, c, \ldots$.
	
	{\bf Lists} are denoted by $U, V, T,X,Y,Z\ldots$.
	By $\el$ we denote the empty list and by
	$a \cons U$ the list having head $a$ and tail $U$.
	The functions \head, \Tail, \last, {\em append} ($\append$), and {\em Concat} ($\concat$) are considered to be executable in one step\footnote{In the background theory \last, {\em append}, and {\em concat} are defined in a functional way as linear time algorithms, however in an efficient array-based computing environment they can be implemented in constant time.}. {$\cons$}, $\el$, and {$\append$} can be used for pattern matching.
    
    {\bf Tuples} are nonempty lists of lists and are denoted using the constructor $\langle \ldots \rangle$ like in $\langle U, V\rangle$. Tuples have only few elements and are used intermediary in some algorithms in order to group few values as a multiple result or a multiple argument. 
	
	{\bf ~Multisets} are denoted by $A, B, C,\ldots$.
	By $\emptyset$ we denote the empty multiset, by $\mse{a,\ldots,b}$ the multiset containing $a,\ldots,b$, and by $\ms[U]$ the multiset containing exactly the elements of the list $U$ with the same number of occurrences. The additive union of $A,B$  is denoted by $A\uplus B$ like in \cite{Knuth:1998:ACP:280635}.
	
	Functions are denoted by special symbols or by names. The names of the functions that return simple domain elements start with lower case letters, while the names of the functions returning lists and multisets start with capital letters.
	
	The total ordering between the elements is extended between individual elements and lists and multisets, as well as between lists and between multisets, by requiring that each element belonging to each argument of the predicate observes the respective relation.
	These orderings are used in proofs of sortedness and should not be confused with the orderings between lists and between multisets that are determined by [strict] inclusion, and which are used for induction.
	
	The definitions of notions and their properties are presented in \tma\ syntax (which is used by the \tma\ system for computations and proving) in a theory file\footnote{The file {\em Theory.nb} in the subdirectory {\em Theorema} of the companion repository.} that we develop for the purpose of studying algorithms on lists.
	Currently the theory contains 5 definitions of predicates (strict inclusion of multisets, ordering of multisets and of lists, list inclusion, and sorting of lists), one definition of a function (the multiset of a list), 14 definitions of algorithms (each with concrete running tests), as well as 38 properties of these notions. These cover 25 pages when printed.
	
	Concerning the \tma\ notations, for function and predicate application we use squared brackets (e.g., $F[x], P[x]$). Quantified variables are written under the quantifier 
	(e.g. $\underset{X }{\forall }$ ``for all $X$'',
	$\underset{X, Y }{\exists }$ ``exists $X, Y$''), and sometimes they are accompanied by a condition (e. g.
	$\cfa{a,b}{a\leq b}$ ``for all $a, b$ such that
	$a\leq b$'' ). The names of the Skolem constants use the name of their ancestor variables, but have integer indices (e. g. $U_0,$ $a_1$).

    In \tma\ the definition of the functions that compose the algorithm are presented as universally quantified logical [conditional] equalities defining [conditional] rewrite rules.
	The arguments of the LHS\footnote{Left-Hand-Side.} may be  variables or {\em list patterns}.
    Pattern matching in algorithm definitions is a specific distinctive feature of \tma\ (inspired by \mma) that simplifies the description of functions but still keeping the semantics of first--order logic. For the algorithms presented here we only use pattern matching on Cons ($\cons$), Append ($\append$), and $\langle\ldots\rangle$. For instance:\\
    \centerline{matching $1\cons(2\cons(3\cons\el))$ with $a\cons U$ gives $a = 1$ and $U = 2\cons( 3\cons\el)$}\\
    \centerline{matching $1\cons(2\cons(3\cons\el))$ with $V\append b$ gives $b = 3$ and $V = 1\cons( 2\cons\el)$}\\
    \centerline{matching $\langle\ 1\cons(2\cons\el),\ 3\cons(4\cons\el)\ \rangle$ with $\langle U, V\rangle$ gives $U = 1\cons(2\cons\el)$ and $V = 3\cons(4\cons\el)$}\\
    
	Each function may have several defining clauses, with the obvious restriction that each possible input matches exactly one clause.
	
	\subsection{The Straightforward Algorithm}
	
	As briefly discussed above, the method consists in scanning the input and distributing the elements between a sorted sublist (by adding the current element at the beginning or at the end of it if it fits) and the rest of the list, which is sorted recursively before the two lists are merged into the final result.
	
	In a functional straightforward way one does not express the preliminary scanning in a single pass, but by using two functions: \selsorted\ constructs the sorted sublist and \selrest\ the other one.
	
	We present below the definitions in \tma\ syntax, but they are practically identical with the ones given in Rocq, with just some syntactic variations related to the two respective language particularities. (In Rocq, additionally, one also has to declare the types of various objects.)
	
	\begin{definition}\label{algorithm:patsort1} {\patsortI.}
		
		\centerline{
			\(
			\begin{array}{c}
				\patsortI[\el]=\el\\
				\fa{a,U}\patsortI[a\cons U] = \merge[\patsortI[\selrest[a\cons U]],\ \selsorted[a\cons U]]
			\end{array}\)
		}
	\end{definition}
	
	\noindent
	\merge\ is the well--known linear algorithm that scans the two sorted lists in parallel and combines them into a sorted list %by choosing the smallest from their heads at each step
	-- see e.g. \cite{Knuth:1998:ACP:280635}.
	
	\noindent
	In the following, \head\ and \last\ have the obvious meaning, while $U\frown a$ denotes append. %ing $a$ to $U$.
	
	\begin{comment}
		\begin{definition}\label{algorithm:selsorted} {\selsorted}
			
			\centerline{
				\(
				\begin{array}{c}
					\selsorted[\el]=\el\\
					\fa{a}\ \selsorted[a\cons \el] = a\cons \el\\
					\cfa{a,U}{a\leq\head[\selsorted[U]]}\ \selsorted[a\cons U] = a\cons \selsorted[U]\\
					\cfa{a,U}{last[\selsorted[U]]\leq a}\ \selsorted[a\cons U] = \selsorted[U]\frown a\\
					\cfa{a,U}{\head[\selsorted[U]]<a<last[\selsorted[U]]} \selsorted[a\cons U] = \selsorted[U]
				\end{array}\)
			}
		\end{definition}
	\end{comment}
	
	\begin{definition}\label{algorithm:selsortedrest} {\selsorted, \selrest.}
		
		\centerline{
			\(
			\begin{array}{c}\\
				\selsorted[\el]=\el\ \wedge\ \selrest[\el]=\el\\ 
				\fa{a}\ \selsorted[a\cons \el] = a\cons \el\ \wedge\ \selrest[a\cons \el] = \el\\
				\cfa{a,U}{a\leq\head[\selsorted[U]]}\hspace{-30pt} \selsorted[a\cons U] = a\cons \selsorted[U]\ \wedge\ \selrest[a\cons U] = \selrest[U]\\
				\cfa{a,U}{last[\selsorted[U]]\leq a}\hspace{-30pt}\selsorted[a\cons U] = \selsorted[U]\frown a\ \wedge\ \selrest[a\cons U] = \selrest[U]\\
				\cfa{a,U}{\head[\selsorted[U]]<a<last[\selsorted[U]]}\hspace{-65pt}\selsorted[a\cons U] = \selsorted[U]\ \wedge\ \selrest[a\cons U] = a\cons\selrest[U]
			\end{array}\)
		}
	\end{definition}
	
	\begin{comment}
		\begin{definition}\label{algorithm:selrest} {\selrest.}
			
			\centerline{
				\(
				\begin{array}{c}
					\selrest[\el]=\el\\
					\fa{a}\ \selrest[a\cons \el] = \el\\
					\cfa{a,U}{a\leq\head[\selrest[U]]}\ \selrest[a\cons U] = \selrest[U]\\
					\cfa{a,U}{last[\selrest[U]]\leq a}\ \selrest[a\cons U] = \selrest[U]\\
					\cfa{a,U}{\head[\selsorted[U]]<a<last[\selsorted[U]]} \selrest[a\cons U] = a\cons\selrest[U]
				\end{array}\)
			}
		\end{definition}
	\end{comment}
	
	\noindent
	Note that, because of the second clause, $\selrest[U]$ is always shorter than $U$,
	thus the function \patsortI\ is terminating.

    A previous version of the certification of this algorithm in \tma\ and Rocq (Coq) have been published in the companion repository of \cite{tokyo}, however we include in the current companion repository an updated version of the proofs.
    
	The multiple recursive calls in the two auxiliary functions make this algorithm very inefficient, however it has a very intuitive definition that expresses naturally the main ideas of the algorithm.
	In the sequel we will illustrate the principles of transformation of this algorithm into an efficient one, namely by changing all the functions that use nested recurrences into ones that use tail-recursion, having in mind the final goal (of future work) that these will provide the basis for efficient array-based programs implemented in imperative languages similar to C.
	
	% All versions of the sorting algorithm use the auxiliary function \merge, which combines two sorted lists into a sorted one.
	
	\begin{comment}
		\begin{definition}\label{algorithm:merge} {\merge.} \todo{Poate renuntam la asta?}
		
		\centerline{
			\(
			\begin{array}{c}
				\merge[\el,\el]=\el\\
				\fa{a,U}\ \merge[a\cons U,\el] = a\cons U\\
				\fa{a,V}\ \merge[\el,b\cons V] = b\cons V\\
				\cfa{a,b,U,V}{a\leq b}\ \merge[a\cons U,b\cons V] = a \cons \merge[U,b\cons V]\\
				\cfa{a,b,U,V}{\neg(a\leq b)}\ \merge[a\cons U,b\cons V] = b \cons \merge[a\cons U, V]\\
			\end{array}\)
		}
	\end{definition}
	\end{comment}

\subsection{Using a Tail-recursive Split}
	
%	The transformation from nested recursion is based on the following main principle:
A recursive definition of a function $F$ has in general 4 types of clauses: trivial clauses, recursion clauses, initialization clauses, and termination clauses. For tail-recursive definitions they have the special structure described below.
\begin{itemize}
	\item The {\em trivial} clauses have the form:\\
    \centerline{$ F[\delta]\ :=\ \epsilon$, where $\delta$ is the input and $\epsilon$ is the output of the function.}
    These clauses are not recursive and they may be missing.
    % In case of transformation of a function that is not tail-recursive, these clauses are imported from the original definition of these one.
    
	\item The {\em recursion} clauses have the form:\\
\centerline{$ F[\alpha_1,\ldots,\alpha_m,\beta_1,\ldots,\beta_n,%
	\gamma_1,\ldots,\gamma_p]\ :=\ 
	F[\alpha'_1,\ldots,\alpha'_m,\beta'_1,\ldots,\beta'_n,%
	\gamma'_1,\ldots,\gamma'_p],$ where:}
\begin{itemize}
	\item the $\alpha$'s denote inputs that are consumed,
	\item the $\beta$'s denote values that are computed in the previous steps and are necessary later, and
	\item the $\gamma$'s denote results that accumulate.
\end{itemize}
On the LHS we may have variables representing domain elements or lists, or pattern expressions (using $\cons$, $\append$, or $\langle\ldots\rangle$) that denote lists or tuples.
On the RHS\footnote{Right-Hand-Side.} we may have any expressions, but the type of expressions on LHS/RHS are in one-to-one correspondence.

	\item The {\em initialization} clauses have the form:
	
\centerline{$ F[\delta]\ :=\ 
	F[\alpha'_1,\ldots,\alpha'_m,\beta'_1,\ldots,\beta'_n,%
	\gamma'_1,\ldots,\gamma'_p],$ where $\delta$ is the input of the function,}
 while the arguments of the RHS have the role presented above.
 These are missing when the function is called from another one, because this call sets up the necessary configuration of the $\alpha/\beta/\gamma$--arguments. 
 	\item The {\em termination} clauses have the form:
 	
 	\centerline{$ F[\alpha_1,\ldots,\alpha_m,\beta_1,\ldots,\beta_n,%
 		\gamma_1,\ldots,\gamma_p]\ :=\ \epsilon,$ where $\epsilon$ is the output of the function,}
 	while the arguments of the LHS have the role presented above.
\end{itemize}

	The new split function follows the scheme above:

	\begin{definition}\label{algorithm:splitsorted} {\patsplit.}
		
		\centerline{
			\(
			\begin{array}{c}
				\fa{a,b,Y,Z}\ \patsplit[a,b,\el,Y,Z] = \langle (a\cons Y)\frown b,Z\rangle\\
				\cfa{a,b,c,X,Y,Z}{c\leq a}\ \patsplit[a,b,c\cons X,Y,Z] = \patsplit[c,b,X,a\cons Y,Z]\\
% 				\cfa{a,b,c,X,Y,Z}{b\leq c\ \wedge\ c > a}\ \patsplit[a,b,c\cons X,Y,Z] = \patsplit[a,c,X,Y\frown b,Z]\\
% 				\cfa{a,b,c,X,Y,Z}{b > c\ \wedge\ c > a}\ \patsplit[a,b,c\cons X,Y,Z] = \patsplit[a,b,X,Y,c\cons Z]
				\cfa{a,b,c,X,Y,Z}{a < c < b}\ \patsplit[a,b,c\cons X,Y,Z] = \patsplit[a,b,X,Y,c\cons Z]\\
    			\cfa{a,b,c,X,Y,Z}{b\leq c}\ \patsplit[a,b,c\cons X,Y,Z] = \patsplit[a,c,X,Y\frown b,Z]
			\end{array}\)
		}
	\end{definition}

    The clauses 2, 3, and 4 are {\em recursion} clauses with arguments $a,b,X,Y,Z$. $a$ and $b$ are the arguments that are computed at previous steps and are necessary later: they contain the minimal, respectively the maximal element of the accumulating result (a sorted sublist of the input), thus they are also part of the sorted sublist from the result.
    $X$ is the rest of the input, which is consumed during the computation by picking-up the head of it at each step (this also insures termination).
    $Y$ and $Z$ are the components of the result: the sorted sublist and the rest of the input, respectively.

    There are no {\em trivial} clauses.

    There are no {\em initialization} clauses, because the set-up of the arguments of the recursion clauses is performed in the clauses 3 and 4 of the the algorithm \patsortII\ (see below).
    Namely these clauses pick-up the first two elements of the input and set them to be the minimal, respectively the maximal elements of the intented result. The rest of the input is set-up as the input argument to be consumed. The components of the result are set to the empty list.

    Clause 1 is the {\em termination} clause: when the input is completely consumed, it returns the tuple containing the sorted sublist (with the extreme elements added) and the rest of the input.
    The use of tuple here is not mandatory\footnote{That is the case in general.}, but it eases the expression of the algorithm as a logical formula, as well as the rendering of the certification proof.
    
	%The core of the algorithm is \patsplit\ that produces a pair $\langle U,V \rangle$ of a sorted sublist and the rest from the input $X$, which is scanned during the execution.
	%The arguments $a,b,Y$ represent the sorted sublist $(a\cons Y)\frown b)$, while $Z$ accumulates the remaining elements.
		
	The main function still remains recursive:
    %, however this does not have a crucial effect on the performance since the depth of the recursion is on average logarithmic with respect to the length of the input.
	
	\bigskip
	\begin{definition}\label{algorithm:patsortII} {\patsortII.}
		
		\centerline{
			\(
			\begin{array}{c}
				\patsortII[\el]=\el\\
				\fa{a}\ \patsortII[a\cons \el] = a\cons \el\\
				\cfa{a,b,X}{a\leq b}\ \patsortII[a\cons (b\cons X)] = \patcomb[\patsplit[a,b,X,\el,\el]]\\
				\cfa{a,b,X}{\neg(a\leq b)}\ \patsortII[a\cons (b\cons X)] = \patcomb[\patsplit[b,a,X,\el,\el]]\\
				\fa{U,V}\patcomb[\langle U,V\rangle]=\merge[U,\patsortII[V]]
			\end{array}\)
		}
	\end{definition}

    \patcomb\ is a simple intermediate function whose role is to distribute the components of the result of \patsplit\ for further processing.
    
	Note that for the termination of the algorithm, $V$ must be strictly smaller than $X$: this is explicitly proven during the certification process.

    \subsection{Tail-recursive Sorting}
    In order to obtain a sorting algorithm that is completely tail-recursive
	we transform the main function using the same general schema as above, namely we express sorting through the tail-recursive function \patjoin:

	\begin{definition}\label{algorithm:patsortIII} {\patsortIII.}
		
		\centerline{
			\(
			\begin{array}{c}
				\patsortIII[X]=\patjoin[\el, X]	
			\end{array}\)
		}
	\end{definition}
	
	\begin{definition}\label{algorithm:patjoin} {\patjoin.}
		
		\centerline{
			\(
			\begin{array}{c}
				\patjoin[X, \el]=X\\
				\fa{a,X}\ \patjoin[X,a\cons \el] = Merge[X, a\cons \el]\\
				\cfa{a,b,X,Y}{a\leq b}\ \patjoin[X,a\cons (b\cons Y)] = \patmix[X,\patsplit[a,b,Y,\el,\el]]\\
				\cfa{a,b,X,Y}{a > b}\ \patjoin[X,a\cons (b\cons Y] = \patmix[X,\patsplit[b,a,Y,\el,\el]]\\
				\fa{X,U,V}\patmix[X,\langle U,V\rangle]=\patjoin[\merge[X,U],V]]
			\end{array}\)
		}
	\end{definition}

    Clauses 3, 4, and 5 encode the {\em recursion}, which has two arguments: the first represents the accumulating sorted result, while the second constitutes the consumed input. The recursion uses the intermediate function \patmix\ whose role is the same as of \patcomb\ of the previous version: it distributes the components of the result of \patsplit\ for further processing. Note that, in contrast to the \patsplit\ recursion, here we do not have a \head/\Tail\ consumation of the input, but a \patsplit\ consumation: at each recursion step the current part of the input is reduced to $V$, the unsorted rest of it after selecting the sorted sublist.
    This means that the termination of the algorithm needs the certification of the fact that $V$ is strictly included in the previous version of the consumed input $a\cons(b\cons Y),$ which is straightforward since $a$ and $b$ are already included in the sorted sublist from the beginning.
    There are no {\em trivial} clauses.
    The {\em initialization} is performed by the call from \patsortIII.
    Clauses 1 and 2 are the {\em termination} ones. 
    
	%\merge\ is still not tail-recursive: a tail-recursive version of it is straightforward, however it would be pretty inefficient in such a functional environment, because it makes heavy use of {\em Append} and {\em Concat}.

    If we consider $\append$ and $\concat$ to be one step functions, then \merge\ can also be expressed in a tail-recursive way, by adding the result as an additional accumulating argument to the function (see the \tma\ implementation in the companion files).
	
	\section{Certification in Rocq and Theorema}
	
	\subsection{The Certification Method}
	
	Since in both  \tma\ and  Rocq the algorithms are expressed as logical formulae, there is no need to generate verification conditions.
	Rather, the function definitions are included in the set of statements that are used for proving the correctness, together with the necessary definitions and properties of the relevant functions and predicates.
	
	These properties constitute a \textit{background theory} of lists and multisets that contains the definitions and properties of all the necessary notions, including ordering and sorting.
	In \tma\ they are explicitly listed in the file {\tt Theory.pdf} present in the companion repository. %\footnote{\url{https://www.risc.jku.at/people/tjebelea/Dramnesc-Jebelean-Stratulat-FROM-2026.html}}.
	In Rocq they are partially created for this paper, and partially included from Rocq libraries.
		
	The correctness is proven with respect to a certain specification consisting of an input condition and an output condition.
	The correctness conjecture states that for any input $X$ satisfying  a certain input condition $P$, the corresponding result of the function $F$ defined by the algorithm satisfies a certain output condition $Q$ with respect to the input.
    
    The final step of the certification consists in proving the two main correctness conjecture, namely that the result of the sorting function is sorted and that it contains exactly the same elements as the input list\footnote{When the proofs of these facts are independent of each other, this is expressed as two separate statements.}.
    In \tma\ the later is formalized as equality between the multiset of the elements of the input and the one of the output.

    A significant difference between the \tma\ and the Rocq versions of the background theory as well as between the certification processes derives from the different approaches to the characterizations of the fact that two lists have the same elements: in \tma\ this is done using multisets, while in Rocq this is done using the {\em permutation} relation.

    It is important to note a significant difference between the correctness proofs of tail-recursive algorithms and the ones of the other ones.
    
    In the case of algorithms that are not tail-recursive, the {\em main term} -- the term consisting in the application of the main function -- represents a value that changes at every recursive call, and the certification statement expresses a certain relation between this term and some other values involved in the computation, some of them also changing.
    Proof of this statement is possible using the changes that are described in the definition of the algorithm.

    In the case of tail-recursive algorithms, the main term represents a value that does not change: it always equals the final result of the computation.
    The changes occur in the values of the arguments of the main term, thus the invariant has to express a relation between these at every recursion.
    However, we want to avoid the use of a meta-logic that would allow to express something like: ``the arguments of the function call during the computation have this property''.
    Our novel method consists in formulating an auxiliary lemma that expresses the fact that, if the arguments of the main term have certain invariant properties, then there exists an element that equals it and has the properties required by the correctness conjecture.
    % This auxiliary lemma is also a kind of invariant property and it facilitates the proof of correctness.
	
	\subsection{Certification in \tma}
	
	\subsubsection{Proving in \tma}
	
	A \tma\ proof is a tree of {\em proof situations} (current goal and assumptions), whose root is the initial statement to be proven, and each other node is obtained by the application of a certain {\em inference rule} to its predecessor.
	The leaves of the tree are established by certain inference rules that decide whether the current node is either a {\em success} or a {\em failure}.
	
	In certain situations it appears to be most effective to apply certain inference rules in a certain order: these situations and the appropriate course of action are described by {\em strategies}.

    For the purpose of this case study we implemented in \tma\ a special prover for lists over an ordered arbitrary domain and multisets, by defining the necessary inference rules and strategies.
	
	\smallskip
	\noindent
	{\bf Inference Rules.}
	The general rules are classical inference rules similar to those of sequent calculus, and other straightforward natural rules like applying properties of equality, using definitions, replacement by equal terms, etc.
	Other important general rules are:
	\begin{itemize}
		\item %\textbf{G1:} 
		\textit{Generalized induction.}
		When proving the goal $\varphi[X]$, one may assume  $\varphi[Y]$ for any $Y$ strictly smaller than $X$ with respect to strict inclusion of multisets ($X\subset Y$).
		
		%\item %\textbf{G2:} 
		%\textit{Cascading.} When the proof fails, construct a conjecture from the current proof state and prove it separately.
		%This step requires the intervention of a human because it is not completely automated.
		
		\item %\textbf{G3:} 
		\textit{Term Contracting.} Replace a repeating composite constant term to a new constant.
		
		%\todo{Example}
	\end{itemize}
	\smallskip
	\noindent
	Additionally, one has inference rules that are specific to the domains of lists and multisets:
	\begin{itemize}
		\item %\textbf{G4:}
		\textit{Equality reduction.} Delete identical terms that occur on both sides of a goal consisting in an equality between two multiset unions.
        
        \item %\textbf{S1:} 
		\textit{List expansion.} When it is assumed that a constant list is nonempty, it is replaced by a term of the form $a\cons Y$, where $a$ and $Y$ are new Skolem constants, and similarly for lists containing at least two elements.
		
		%\item %\textbf{S2:} this is almost micro-terms
		%\textit{Multiset expansion.} The term $\ms[T]$ where $T$ is a composite term representing a list can be expanded into several $\ms$ terms by using multiset union. This is used mostly in conjunction with the equality reduction.
		
		\item %\textbf{S3:}
		\textit{Inequality reduction.} When the goal is an inclusion or and inequality (with respect to the order used for sorting) between terms representing lists or multisets, a subterm $T$ on the LHS may be replaced by a term that is known to be included in $T$, and a subterm on the RHS may be replaced by a term that it is known to include it.
	\end{itemize}
	\smallskip
	\noindent
	The use of inference rules is sometimes triggered by {\em strategies}, from which the most important are:
	\begin{itemize}
		\item {\em Lemma for invariant.} Generalize the correctness statement of a tail-recursive algorithm into a lemma that can be proven to hold at each recursive call -- see lemmas \ref{lemma:PatSplit-correct} and \ref{lemma:PatSplit-Invariant}.
        
		\item {\em Micro-atoms.} Transform an atom\footnote{An atom is a subformula consisting of a predicate symbol applied to some terms.} into a conjunction of several ``micro-atoms'', that is terms whose arguments are not composite.\\
		% Typically used for ordering predicates and for the sorting predicate \issorted.
		{Example:} $\issorted[a\cons U]$ decomposes into $(a\leq U)\wedge \issorted[U]$.
        \item {\em Micro-terms.} Transform a term into several ``micro-terms'', that is terms whose arguments are not composite.\\
        %Typically used for multiset union.% (see also rule \textbf{S2} above).
		{Example:} $\ms[a\cons U]$ is transformed into $\mse{a}\uplus\ms[U]$
		\item {\em Proof by algorithm cases.} For proving a statement about a function, create a branch for each clause of the respective algorithm.
	\end{itemize}
	
	\subsubsection{Certification Statements and Proofs}

    The certification process consists in defining and proving several statements, most of them being lemmas that are used in the proof of the main correctness conjectures.

	{\bf PatSort1} has two certification statements:
	%for each algorithm, namely the preservation of the multiset and the fact that the output is sorted. For instance:
	\begin{theorem}{PatSort1-Preserves.}\hspace{30pt}$\fa{X}{\ms[PatSort1[X]]=\ms[X]}$
	\end{theorem}
    
	\begin{theorem}{PatSort1-Sorted.}\hspace{55pt}$\fa{X}{\issorted[PatSort1[X]]}$	
	\end{theorem}
	\noindent
    Both for \patsortI\ and for \merge\ (see below), the proof of the sorting property needs the multiset preservation statement. \merge\ has a different structure, thus its certification statements are:
	\begin{lemma}{Merge-Preserves.}\hspace{30pt}$\fa{X,Y}{\ms[\merge[X,Y]]=\ms[X]\uplus\ms[Y]}$\end{lemma}
    
	\begin{lemma}{Merge-Sorted.}\hspace{10pt}$\fa{X,Y}{(\issorted[X]\wedge\issorted[Y])\implies\issorted[\merge[X,Y]]}$\end{lemma}
    
	The proofs of these two statements are by cases, one for each clause of the definition of the function.
	For the cases where the definition involves recursion, the proof proceeds by induction: the induction hypothesis is generated according to the recursive term in the definition.
	For instance, in the case corresponding to the last clause of \merge\ in the proof of multiset preservation one has to prove for arbitrary but fixed $a_0,b_0,U_0,V_0$: $\issorted[\merge[a_0\cons U_0,b_0\cons V_0]]$.
	The generated induction hypothesis is: $\issorted[\merge[a_0\cons U_0, V_0]]$.
	The induction scheme is sound because the union of multisets of the arguments of the function in the induction hypothesis is strictly included in the union of multisets of the arguments of the function in the induction conclusion.
	This follows from a rule used by the prover: the multiset of symbols\footnote{This is the multiset of symbols at meta level, it does not denote an term at object level.} $\mse{a_0,U_0,V_0}$ is strictly included in the multiset of symbols $\mse{a_0,b_0,U_0,V_0}$.
	This kind of induction is used in all the proofs regarding the functions that are defined by structural induction.
	
	In contrast, the definition of \patsortI\ on $a\cons U$ uses a recursive call to $\selrest[a\cons U]$, thus the
	correctness statement also requires that $\selrest[X]$ is always strictly included in $X$, which insures that the algorithm terminates, and simultaneously that the induction scheme used in the proofs of its correctness statements is sound.
	This strict inclusion is proven as separate lemma, as it is also done for other functions that are not defined by structural induction, see below.
	
	\smallskip
	\noindent
	{\bf PatSort2} has only one certification statement, this is due to the fact that two separate proofs of the two assertions would have many identical steps.
	\begin{theorem}{PatSort2-correct.}\ $\fa{X}{\ (\issorted[\patsortII[X]]\ \wedge\ \ms[\patsortII[X]]=\ms[X]} )$\end{theorem}
	\noindent 
	The induction scheme of the proof uses the fact that the multiset of a certain term (in this case $V_0$) that is produced during the proof is strictly included in the multiset of the term of the induction conclusion (in this case $a_0\cons (b_0\cons X_0)$). Therefore the induction hypothesis is produced only after $V_0$ is produced and the fact that its multiset is smaller is proven.
	The same strategy is used in other proofs also, see below.
	
	The proof uses the correctness of \patsplit, which includes both the multiset preservation and the sorting properties, and additionally the non-emptiness of the first component of the result, which is necessary for the soundness of the induction mentioned above.
	Moreover, in order to avoid an index-based access\footnote{This would make the formulation of the lemma much more complex and less intuitive.} to the elements of the pair that is produced by the function, we express the correctness using an existential statement.
	\begin{lemma}{PatSplit-correct.}\label{lemma:PatSplit-correct}
		
		\centerline{
			\(
			\underset{a\leq b}{\underset{a,b,X}{\forall}}\ 
			\underset{U,V}{\exists}
			%\left{
				\begin{array}{l}
					\patsplit[a,b,X,\el,\el]=\langle U,V\rangle\ \wedge
					\issorted[U]\ \wedge\ U\not{=}\el\ \wedge\\
					\ms[a\cons(b\cons X)]=\ms[U]\uplus\ms[V]
				\end{array}
				%\right.
				\)
			}
		\end{lemma}
		
		As it is typical for tail-recursive functions, the correctness of \patsplit\ requires another lemma (the invariant), which is possible to prove by induction on the recursive step of the algorithm.
		
		\begin{lemma}{PatSplit-Invariant.}\label{lemma:PatSplit-Invariant}
			
			\centerline{
				\(
				\underset{X}{\forall}\ 
				\underset{\issorted[(a\cons Y)\frown b]}{\underset{a,b,Y,Z}{\forall}}\ 
				\underset{U,V}{\exists}
				%\left{
					\begin{array}{l}
						\patsplit[a,b,X,Y,Z]=\langle U,V\rangle\ \wedge
						\issorted[U]\ \wedge\ U\not{=}\el\ \wedge\\
						\ms[a\cons(b\cons X)]\uplus\ms[Y]\uplus\ms[Z]=\ms[U]\uplus\ms[V]
					\end{array}
					%\right.
					\)
				}
			\end{lemma}
			Since the function is defined by a simple head-tail recursion on $X$, the proof works by a simple head-tail induction on $X$, it is however the largest proof of the case study (55 steps).
			
			\smallskip
			\noindent
			{\bf \patsortIII}\ has a very simple correctness proof because it reduces immediately to the correctness of \patjoin, both having the two correctness properties in the same statement.
			The latter also plays the role of an invariant lemma, by generalizing the first argument:
			\begin{lemma}{PatJoin-correct.}
				
				\centerline{
					\(
					\begin{array}{c}
						\cfa{X,Z}{\issorted[X]}{( 
							\issorted[\patjoin[X,Z]] \ \wedge\ \ms[\patjoin[X,Z]]=\ms[X]\uplus\ms[Z]} )
					\end{array}\)
			}\end{lemma}
			
			The induction scheme of the proof is based of the recursive definition of the function, which is essentially:
			$$\patjoin[X,Z] = \patjoin[\merge[X,U],V],$$
			where $U$ and $V$ are produced by \patsplit\ from $Z,$
			thus $\ms[V]$ is strictly included in $\ms[Z].$
			
			Therefore we first isolate $Z$ in a prefix universal quantifier and then we prove this formula by well-founded induction using an arbitrary but fixed $Z_0$, as well as the induction hypothesis:
			
			$$\cfa{V}{\ms[V]\subset\ms[Z_0]}{\ \cfa{X}{\issorted[X]}{ (\issorted[\patjoin[X,V]]\ \wedge\ 
					(\ms[\patjoin[X,V]]=\ms[X]\uplus\ms[V])}   
			}$$
			and the induction conclusion:
			$$\cfa{X}{\issorted[X]}{ (\issorted[\patjoin[X,Z_0]]\ \wedge\ 
				(\ms[\patjoin[X,Z_0]]=\ms[X]\uplus\ms[Z_0])}   
			$$
			During the proof, after $V_0$ is produced and one proves $\ms[V_0]\subset\ms[Z_0],$ then the induction hypothesis is instantiated to $V_0$, which makes the proof succeed.
			
			\subsection{Certifying Sorting Algorithms in Rocq}
			
			\subsubsection{Specifying the Properties of Sorting Algorithms}
			
			For the sake of simplicity, the sorting algorithms will be represented in Rocq as functions that take as input a list $l$ of naturals and return a permutation of $l$ that is sorted in increasing order. 
			
			The datatypes for naturals and lists are inductively defined, based on \textit{constructors}. The constructors for 
			\begin{itemize}
				\item naturals, of type \coqdocvar{nat}, are
				\coqdocvar{0} and \coqdocvar{S} (the successor function), and
				\item lists, of type \coqdocvar{list}, are
				\coqdocvar{nil} and \coqdocvar{cons}. Instead of \coqdocvar{cons}, we will use the {::} infix notation.
			\end{itemize}
			
			Rocq requires that every function be \textit{total}. It uses the \coqdockw{Fixpoint} keyword for recursive functions whose termination can be automatically proved. The functions map inputs to outputs using \textit{match} constructs representing sequences of the \coqdockw{match} keyword, an expression $e$, the \texttt{with} keyword, the match cases and the \texttt{end} keyword. The match cases should correspond to the constructor cases of the datatype of $e$. The output for each case is prefixed by the \ensuremath{\Rightarrow}  symbol and may contain other \emph{match} constructs to refine the behaviour of the function. The match-based approach allows Rocq to identify and display the missing cases.
			
			As an example, we can define the recursive function \coqdocvar{count}  that counts the number of occurrences of an element \coqdocvar{e} in a list \coqdocvar{lst}, as follows :
			
			\coqdocemptyline
			\coqdocnoindent
			\coqdockw{Fixpoint} \coqdocvar{count} \coqdocvar{e} \coqdocvar{lst} :=\coqdoceol
			\coqdocindent{1.00em}
			\coqdockw{match} \coqdocvar{lst} \coqdockw{with}\coqdoceol
			\coqdocindent{1.60em}
			\coqdocvar{nil} \ensuremath{\Rightarrow} 0\coqdoceol
			\coqdocindent{1.00em}
			\ensuremath{|} \coqdocvar{head} :: \coqdocvar{tail} \ensuremath{\Rightarrow} 
			\coqdockw{if} \coqdocvar{e} =? \coqdocvar{head} \coqdockw{then} \coqdocvar{S} (\coqdocvar{count} \coqdocvar{e} \coqdocvar{tail}) \coqdockw{else} \coqdocvar{count} \coqdocvar{e} \coqdocvar{tail}\coqdoceol
			\coqdocindent{1.00em}
			\coqdockw{end}.\coqdoceol
			\coqdocemptyline
			
			\noindent
			where \coqdocvar{=?} denotes the boolean equality.\\
			
			Recursive functions should terminate. Rocq can automatically check the termination property for a function $f$ if there is one of the arguments of $f$ that structurally decreases after each recursive call, as is, for example, the second argument of the  \coqdocvar{count} function.
			
			We are ready to define the permutation relation between two lists which holds if, for every element $x$, $x$ occurs the same number of times in both lists. When the input list does not have repeated elements, two lists are in a permutation relation if whenever an element occurs in one of the lists it also occurs in the other list, and viceversa. In this case, it is preferable to use the \coqdocvar{In} membership predicate, for which \coqdocvar{In} \coqdocvar{e} \coqdocvar{lst}  returns true if \coqdocvar{e} occurs in \coqdocvar{lst}. This is because the underlying reasoning is simpler and does not involve arithmetics. To take into account the two cases, our definition of permutation is the conjunction of the \coqdocvar{count}- and \coqdocvar{In}-based definitions:
			
			\coqdocemptyline
			\coqdocnoindent
			\coqdockw{Definition} \coqdocvar{permutation} \coqdocvar{lst} \coqdocvar{lst'} := 
			\coqdockw{\ensuremath{\forall}} \coqdocvar{e},  (\coqdocvar{count} \coqdocvar{e} \coqdocvar{lst} = \coqdocvar{count} \coqdocvar{e} \coqdocvar{lst'}) \ensuremath{\land} (\coqdocvar{In} \coqdocvar{e} \coqdocvar{lst} \ensuremath{\leftrightarrow} \coqdocvar{In} \coqdocvar{e} \coqdocvar{lst'}).\coqdoceol
			\coqdocemptyline
			
			The \coqdocvar{IsSorted} predicate, which tests whether a list is sorted in increasing order, can be defined inductively by using the keyword \coqdockw{Inductive}:  
			
			\coqdocemptyline
			\coqdocnoindent
			\coqdockw{Inductive} \coqdocvar{IsSorted} : \coqdocvar{list} \coqdocvar{nat} \ensuremath{\rightarrow} \coqdockw{Prop} :=\coqdoceol
			\coqdocindent{1.00em}
			\coqdocvar{snil} : \coqdocvar{IsSorted} \coqdocvar{nil}\coqdoceol
			\coqdocnoindent
			\ensuremath{|} \coqdocvar{s1} : \coqdockw{\ensuremath{\forall}} \coqdocvar{x}, \coqdocvar{IsSorted} (\coqdocvar{x}::\coqdocvar{nil})\coqdoceol
			\coqdocnoindent
			\ensuremath{|} \coqdocvar{s2} : \coqdockw{\ensuremath{\forall}} \coqdocvar{x} \coqdocvar{y} \coqdocvar{l}, \coqdocvar{IsSorted} (\coqdocvar{y}::\coqdocvar{l}) \ensuremath{\rightarrow} \coqdocvar{x} \ensuremath{\le} \coqdocvar{y} \ensuremath{\rightarrow} 
			\coqdocvar{IsSorted} (\coqdocvar{x}::\coqdocvar{y}::\coqdocvar{l}).\coqdoceol
			\coqdocemptyline

			Since a sorting algorithm can be any function $f$ that transforms a list of naturals into a permutation of it that is increasingly sorted, the soundness of $f$ can be defined as follows:
			
			\coqdocemptyline
			\coqdocnoindent
			\coqdockw{Definition} \coqdocvar{is\_a\_sorting\_algorithm} (\coqdocvar{f}: \coqdocvar{list} \coqdocvar{nat} \ensuremath{\rightarrow} \coqdocvar{list} \coqdocvar{nat}) 
			
			\hfill := 
			\coqdockw{\ensuremath{\forall}} \coqdocvar{l}, \coqdocvar{permutation} (\coqdocvar{f} \coqdocvar{l}) \coqdocvar{l} \ensuremath{\land} \coqdocvar{IsSorted} (\coqdocvar{f} \coqdocvar{l}).\coqdoceol
			\coqdocemptyline
			
			\noindent 
			
			Given a sorting algorithm $f$, the corresponding soundness conjecture is:
			
			\coqdocemptyline
			\coqdocnoindent
			\coqdockw{Theorem} \coqdocvar{f\_is\_correct}:
			\coqdocvar{is\_a\_sorting\_algorithm} \coqdocvar{f}.\coqdoceol
			\coqdocemptyline
			
			The proof of this theorem will consist in two parts: i) the `permutation' part proving that  \ensuremath{\forall} \coqdocvar{l}, \coqdocvar{permutation} (\coqdocvar{f} \coqdocvar{l}) \coqdocvar{l}, and ii) the `sorting' part proving that \ensuremath{\forall} \coqdocvar{l}, \coqdocvar{IsSorted} (\coqdocvar{f} \coqdocvar{l}).
			
			\centerline{}
			
			In the following, we show how to specify \coqdocvar{PatSort3} in Rocq and prove its soundness.  For lack of space, we present only the specification and the main lemmas needed to certify \coqdocvar{PatSort3}. The repository companion  contains the full specifications and the proofs for \coqdocvar{PatSort1}, \coqdocvar{PatSort2}, and \coqdocvar{PatSort3} using Rocq.
			
			\subsubsection{Specifying and Certifying \coqdocvar{PatSort3}}

			The function \coqdocvar{PatSplit} is defined as follows:\\

			% \coqdocnoindent
			% \coqdockw{Lemma} \coqdocvar{sorted\_conc}: \coqdockw{\ensuremath{\forall}} \coqdocvar{l1} \coqdocvar{l2}, \coqdocvar{sorted} (\coqdocvar{l1} ++ \coqdocvar{l2}) \ensuremath{\rightarrow} (\coqdocvar{sorted} \coqdocvar{l1} \ensuremath{\land} \coqdocvar{sorted} \coqdocvar{l2}).\coqdoceol
			%  \coqdocemptyline
			
			% \coqdocnoindent
			% \coqdockw{Lemma} \coqdocvar{Permutation\_permutation}: \coqdockw{\ensuremath{\forall}} \coqdocvar{l1} \coqdocvar{l2}, \coqdocvar{permutation} \coqdocvar{l1} \coqdocvar{l2} \ensuremath{\leftrightarrow} \coqdocvar{Permutation} \coqdocvar{l1} \coqdocvar{l2} .\coqdoceol
			%  \coqdocemptyline
			
			% \coqdocnoindent
			% \coqdockw{Lemma} \coqdocvar{permutation\_Merge} : \coqdockw{\ensuremath{\forall}} \coqdocvar{l1} \coqdocvar{l2}, \coqdocvar{Permutation} (\coqdocvar{Merge} (\coqdocvar{l1}, \coqdocvar{l2})) (\coqdocvar{l1} ++ \coqdocvar{l2}).\coqdoceol
			% \coqdocemptyline
			\coqdocnoindent
			\coqdockw{Fixpoint} \coqdocvar{PatSplit} \coqdocvar{a} \coqdocvar{b} \coqdocvar{l} \coqdocvar{Y} \coqdocvar{Z} :=\coqdoceol
			\coqdocindent{1.00em}
			\coqdockw{match} \coqdocvar{l} \coqdockw{with}\coqdoceol
			\coqdocindent{2.00em}
			\coqdocvar{nil} \ensuremath{\Rightarrow} ((\coqdocvar{a}::\coqdocvar{Y}) ++ [\coqdocvar{b}], \coqdocvar{Z})\coqdoceol
			\coqdocindent{1.00em}
			\ensuremath{|} \coqdocvar{c} :: \coqdocvar{X} \ensuremath{\Rightarrow} \coqdockw{if} \coqdocvar{c} <=? \coqdocvar{a} \coqdockw{then} 
			\coqdocvar{PatSplit} \coqdocvar{c} \coqdocvar{b} \coqdocvar{X} (\coqdocvar{a}::\coqdocvar{Y}) \coqdocvar{Z}\coqdoceol
			\coqdocindent{6em}
			\coqdockw{else}
			\coqdockw{if} \coqdocvar{b} <=? \coqdocvar{c} \coqdockw{then} 
			\coqdocvar{PatSplit} \coqdocvar{a} \coqdocvar{c} \coqdocvar{X} (\coqdocvar{Y} ++ [\coqdocvar{b}]) \coqdocvar{Z}\coqdoceol
			\coqdocindent{8.5em}
			\coqdockw{else}
			\coqdocvar{PatSplit} \coqdocvar{a} \coqdocvar{b} \coqdocvar{X} \coqdocvar{Y} (\coqdocvar{c}::\coqdocvar{Z})\coqdoceol
			\coqdocindent{1.00em}
			\coqdockw{end}.\coqdoceol
			\coqdocemptyline
			\coqdocnoindent

			\noindent
			where ++ (<=?) is the infix notation of the list concatenation (boolean `less-or-equal', resp.). \\
			
			The definition of \coqdocvar{PatSort3} is based on the  auxiliary recursive function \coqdocvar{PatJoin}. Since its termination cannot be proved automatically, the use of the  \coqdockw{Function} keyword leaves the possibility to the user to specify the decreasing argument and the well-founded ordering, then to prove that the argument is indeed decreasing w.r.t the well-founded ordering in the recursive calls:\\
			
			\coqdocemptyline
			\coqdocnoindent
			\coqdockw{Function} \coqdocvar{PatJoin} (\coqdocvar{X}:\coqdocvar{list} \coqdocvar{nat}) \coqdocvar{l} \{\coqdockw{wf} (\coqdockw{fun} (\coqdocvar{l1}: \coqdocvar{list} \coqdocvar{nat}) (\coqdocvar{l2}: \coqdocvar{list} \coqdocvar{nat}) \ensuremath{\Rightarrow} \coqdoceol
			\coqdocindent{24.00em} \coqdocvar{Nat.lt} (\coqdocvar{length} \coqdocvar{l1}) (\coqdocvar{length} \coqdocvar{l2})) \coqdocvar{l}\} :=\coqdoceol
			\coqdocindent{1.00em}
			\coqdockw{match} \coqdocvar{l} \coqdockw{with} \coqdoceol
			\coqdocindent{2.00em}
			\coqdocvar{nil} \ensuremath{\Rightarrow} \coqdocvar{X}\coqdoceol
			\coqdocindent{1.00em}
			\ensuremath{|} \coqdocvar{a} :: \coqdocvar{l'} \ensuremath{\Rightarrow} \coqdoceol
			\coqdocindent{3.00em}
			\coqdockw{match} \coqdocvar{l'} \coqdockw{with} \coqdoceol
			\coqdocindent{4.00em}
			\coqdocvar{nil} \ensuremath{\Rightarrow} \coqdocvar{Merge} (\coqdocvar{X}, [\coqdocvar{a}])\coqdoceol
			\coqdocindent{3.00em}
			\ensuremath{|} \coqdocvar{b} :: \coqdocvar{Y}  \ensuremath{\Rightarrow} \coqdoceol
			\coqdocindent{6.00em}
			\coqdockw{let} (\coqdocvar{U},\coqdocvar{V}) := \coqdoceol
			\coqdocindent{7.00em}
			\coqdockw{if} \coqdocvar{a} <=? \coqdocvar{b} \coqdockw{then} \coqdocvar{PatSplit} \coqdocvar{a} \coqdocvar{b} \coqdocvar{Y} \coqdocvar{nil} \coqdocvar{nil} \coqdoceol
			\coqdocindent{7.00em}
			\coqdockw{else} \coqdocvar{PatSplit} \coqdocvar{b} \coqdocvar{a} \coqdocvar{Y} \coqdocvar{nil} \coqdocvar{nil} \coqdoceol
			\coqdocindent{6.00em}
			\coqdoctac{in} \coqdocvar{PatJoin} (\coqdocvar{Merge} (\coqdocvar{X}, \coqdocvar{U})) \coqdocvar{V}\coqdoceol
			\coqdocindent{3.00em}
			\coqdockw{end}\coqdoceol
			\coqdocindent{1.00em}
			\coqdockw{end}.\coqdoceol
			\coqdocemptyline
			The \coqdocvar{PatJoin} function satisfies the following crucial properties:
			\coqdocemptyline
			\coqdocnoindent
			\coqdockw{Lemma} \coqdocvar{permutation\_PatJoin}: \coqdockw{\ensuremath{\forall}} \coqdocvar{U} \coqdocvar{V}, \coqdocvar{permutation} (\coqdocvar{PatJoin} \coqdocvar{U} \coqdocvar{V}) (\coqdocvar{U} ++ \coqdocvar{V}).\coqdoceol
			\coqdocemptyline
			\coqdocnoindent
			\coqdockw{Lemma} \coqdocvar{sorted\_PatJoin}: \coqdockw{\ensuremath{\forall}} \coqdocvar{U} \coqdocvar{V}, \coqdocvar{IsSorted} \coqdocvar{U} \ensuremath{\rightarrow}  \coqdocvar{IsSorted} (\coqdocvar{PatJoin} \coqdocvar{U} \coqdocvar{V}).\coqdoceol
			
			\coqdocemptyline
			The \coqdocvar{PatSort3} function is a particular case of \coqdocvar{PatJoin} when the first argument is \coqdocvar{nil}:
			
			\coqdocemptyline
			\coqdocnoindent
			\coqdockw{Definition} \coqdocvar{PatSort3} \coqdocvar{l} := \coqdocvar{PatJoin} \coqdocvar{nil} \coqdocvar{l}.\coqdoceol
			\coqdocemptyline
			
			Finally, the proof of the main theorem \coqdocvar{PatSort3\_is\_correct}, stating that \coqdocvar{is\_a\_sorting\_algorithm} \coqdocvar{PatSort3} holds and specified in Appendix~\ref{sec:lemmas}, is based on the following lemmas:
			
			\coqdocemptyline
			\coqdocnoindent
			\coqdockw{Lemma} \coqdocvar{PatSort3\_permutation} : \coqdockw{\ensuremath{\forall}} \coqdocvar{l}, \coqdocvar{permutation} (\coqdocvar{PatSort3} \coqdocvar{l}) \coqdocvar{l}.\coqdoceol
			\coqdocemptyline
			\coqdocnoindent
			\coqdockw{Lemma} \coqdocvar{PatSort3\_is\_sorted} : \coqdockw{\ensuremath{\forall}} \coqdocvar{l}, \coqdocvar{IsSorted} (\coqdocvar{PatSort3} \coqdocvar{l}).\coqdoceol
			\coqdocemptyline
			
			\centerline{}
			
			The specification of the \coqdockw{Merge} function and the auxiliary lemmas used in the Rocq proof are also given in Appendix~\ref{sec:lemmas}.
			
			\centerline{}
			
			The specifications and the proofs used the modules \coqdockw{Arith}, \coqdockw{List}, and \coqdockw{Lia} (a decision procedure for Linear Integer Arithmetic) from the \coqdockw{Stdlib} library. Most of the proofs used induction reasoning. The induction schemas were built by Rocq from the definition of \coqdockw{List} and from the defined recursive functions \coqdockw{Merge} and \coqdockw{PatJoin}. For the last case, we have used the modules   \coqdockw{FunInd}, \coqdockw{Recdef} and \coqdockw{Wellfounded} from \coqdockw{Stdlib}. For example, the functional scheme associated to \coqdockw{PatJoin} can be defined as: 
			
			\coqdocemptyline
			\coqdocnoindent
			\coqdockw{Functional Scheme} \coqdocvar{PatJoin\_ind} := \coqdockw{Induction} \coqdockw{for} \coqdocvar{PatJoin} \coqdockw{Sort} \coqdockw{Prop}.\coqdoceol
			\coqdocemptyline
			
			\noindent
			and used in the proof by the command \coqdockw{apply} \coqdocvar{PatJoin\_ind}.

			\section{Conclusions and Further Work}
			\subsection{Comparison between Theorema and Rocq} 
			The main differences between the two formal systems appear in algorithm definition, proof development, and proof presentation.
			We detail each aspect below. \\
			
			\noindent\textbf{Algorithm definition.} The types in \tma\ are not explicitly declared, but they are based on notation conventions.
			\tma\ can accept partial functions, such as the ones
			that are not defined when the argument is an empty list. On the other
			hand, the Rocq specifications are typed and accept only total
			functions. In Rocq, partial functions can be handled using the \texttt{option} type~\cite{tokyo}.
			
			Another important difference between systems is the following: in Rocq one can only define terminating functions, while in \tma\  this is not necessary.
			
			In more detail about the concrete approach, in \tma\ one uses multisets in order to compare the contents of lists, while in Rocq this is done using the occurrence count. In Rocq, one could also consider the permutation definition given in the \texttt{Sorting.Permutation} library, by taking advantage of the companion theorems, to ease the proof effort. The use of parameterized lists, with  elements belonging to some totally ordered domain, instead of lists of naturals is also possible. 
            
            Both in \tma\ and Rocq, the certification process required only the definition of the algorithms.
            The \tma\ proofs use some of the functions, predicates, and their properties that are specified in the background theory file, most of them being developed on the occasion of this case study. These are however not specific to the particular algorithms that are investigated here, but are referring in general to ordered domains and multisets and lists over them.
            The Rocq certification is in a similar situation as it uses the appropriate definitions and properties from the relevant libraries related to lists and natural numbers.
			
			\noindent\textbf{Proof development.} The proofs in both Rocq and \tma\ systems have been built using the \emph{cascading} proof strategy: when a proof fails on some subgoal, this is converted to a lemma whose
			proof may be based on other lemmas that are expected to be simpler to prove or are already proved.
            In Rocq some lemmas can be directly accessed from the standard library that comes with the system.
            However, many necessary lemmas had to be created by the user.
			In \tma\ the certification process does not need a specific domain of ordered elements, but we use inference rules which are based on the natural properties of total order (e.g. transitivity). 
            However, the test of the algorithms uses natural numbers as implemented in \textit{Mathematica}.
            
			In Rocq, once a proofs is finished, it is certified by the kernel of the system using the Curry-Howard isomorphism, therefore the proofs are automatically ensured to be logically correct.
            In contrast, in \tma\ it is possible to introduce inference rules that are not logically correct, therefore one has to take much care when designing them.
            
            In Rocq, the proof effort consisted in producing the proof of 22 lemmas and one theorem. Its final version required more than 2000 user interactions\footnote{A Rocq interaction consists essentially in providing a keyword that will trigger the execution of some inference steps.}.
            The \tma\ proofs total a number of 295 inference steps.
            %\marginpar{to check for Theorema.}\\
			
			\noindent\textbf{Proof presentation.}
			The proof presentation is completely different in the two systems.
			
			The \tma\ proofs are generated in natural style, are easy to read as they are similar to human proofs. In contrast, in Rocq the user needs the computer to run the proof scripts step by step which displays the current state of the proof. The generation of Rocq scripts was highly interactive. However, Rocq gives the opportunity to the user to automate the proof process by using tactics and solvers for decision problems.
			
			The Rocq system is producing proofs that are not intended for human reading, but only for further automatic checking.
			Therefore what the user can see is just the produced script.
			
			\subsection{Future Work}
			
            Based on the iterative improvement of the Patience Sort algorithm presented here, one can study the possibility of finding precise rules for this process, finally leading to the automatic algorithm transformation methods.
			
			Another interesting possible follow-up of this case study is the systematic investigation of the principles that are used in creating the lemmas necessary for the proofs, in order to automate this creative process.
			
			The parallel approach to these proofs in both systems constitute a case study for further investigation of possible interactions between the two systems, as:
			\begin{itemize}
				\item the transformation of Rocq proof scripts in \tma-like natural-style proofs,
				\item the generation of Rocq scripts by the \tma\ system,
				\item the verification of \tma\ proofs using Rocq.
			\end{itemize}

\section*{Acknowledgements}

This work is co-funded by the EU through the Erasmus+ project \emph{AiRobo: Artificial Intelligence-based Robotics}, 2023-1-RO01-KA220-HED-000152418.

			\bibliographystyle{eptcs}
			\bibliography{WUT-biblio,
				WUT-biblio-ARC.bib,
				UL-biblio.bib}
			%\bibliography{lipics-v2021-sample-article}
			
			\newpage
			\appendix
			\section{Other functions and lemmas used in the Rocq proof}
			\label{sec:lemmas}
			
			\coqdocemptyline
			\coqdocnoindent
			\coqdockw{Lemma} \coqdocvar{sorted\_sorted}: \coqdockw{\ensuremath{\forall}} \coqdocvar{a} \coqdocvar{l}, \coqdocvar{IsSorted} (\coqdocvar{a}::\coqdocvar{l}) \ensuremath{\rightarrow} \coqdocvar{IsSorted} \coqdocvar{l}.\coqdoceol
			\coqdocemptyline
			\coqdocnoindent
			\coqdockw{Lemma}  \coqdocvar{sorted\_sorted\_rev}: \coqdockw{\ensuremath{\forall}} \coqdocvar{a} \coqdocvar{l}, \coqdocvar{IsSorted} \coqdocvar{l}  \ensuremath{\rightarrow} (\coqdockw{\ensuremath{\forall}} \coqdocvar{y}, \coqdocvar{In} \coqdocvar{y} \coqdocvar{l} \ensuremath{\rightarrow} \coqdocvar{a} \ensuremath{\le} \coqdocvar{y}) \ensuremath{\rightarrow}  \coqdocvar{IsSorted} (\coqdocvar{a}::\coqdocvar{l}).\coqdoceol
			\coqdocemptyline
			\coqdocnoindent
			\coqdockw{Lemma} \coqdocvar{sorted\_app} : \coqdockw{\ensuremath{\forall}} \coqdocvar{y} \coqdocvar{l}, (\coqdockw{\ensuremath{\forall}} \coqdocvar{x}, \coqdocvar{In} \coqdocvar{x} \coqdocvar{l} \ensuremath{\rightarrow} \coqdocvar{Nat.le} \coqdocvar{x} \coqdocvar{y}) \ensuremath{\rightarrow} \coqdocvar{IsSorted} (\coqdocvar{l}) \ensuremath{\rightarrow} \coqdocvar{IsSorted} (\coqdocvar{l} ++ [\coqdocvar{y}]).\coqdoceol
			\coqdocemptyline
			\coqdocnoindent
			\coqdockw{Lemma} \coqdocvar{sorted\_conc}: \coqdockw{\ensuremath{\forall}} \coqdocvar{l1} \coqdocvar{l2}, \coqdocvar{IsSorted} (\coqdocvar{l1} ++ \coqdocvar{l2}) \ensuremath{\rightarrow} (\coqdocvar{IsSorted} \coqdocvar{l1} \ensuremath{\land} \coqdocvar{IsSorted} \coqdocvar{l2}).\coqdoceol
			\coqdocemptyline
			\coqdocnoindent
			\coqdockw{Lemma} \coqdocvar{count\_app} : \coqdockw{\ensuremath{\forall}} \coqdocvar{x} \coqdocvar{l1} \coqdocvar{l2}, \coqdocvar{count} \coqdocvar{x} (\coqdocvar{app} \coqdocvar{l1} \coqdocvar{l2}) = \coqdocvar{count} \coqdocvar{x} \coqdocvar{l1} + (\coqdocvar{count} \coqdocvar{x} \coqdocvar{l2}).\coqdoceol
			\coqdocemptyline
			\coqdocnoindent
			\coqdockw{Function} \coqdocvar{Merge} \coqdocvar{l} \{\coqdockw{wf} (\coqdockw{fun} (\coqdocvar{l}:(\coqdocvar{list} \coqdocvar{nat} \ensuremath{\times} \coqdocvar{list} \coqdocvar{nat})) \coqdoceol
			\coqdocindent{13.00em}
			(\coqdocvar{l'}:(\coqdocvar{list} \coqdocvar{nat} \ensuremath{\times} \coqdocvar{list} \coqdocvar{nat}))  \ensuremath{\Rightarrow} \coqdoceol
			\coqdocindent{12.00em}
			\coqdocvar{Nat.lt} ((\coqdocvar{length} (\coqdocvar{fst} \coqdocvar{l})) + (\coqdocvar{length} (\coqdocvar{snd} \coqdocvar{l}))) \coqdoceol
			\coqdocindent{13.00em}
			((\coqdocvar{length} (\coqdocvar{fst} \coqdocvar{l'})) + \coqdocvar{length} (\coqdocvar{snd} \coqdocvar{l'}))) \coqdocvar{l}\} := \coqdoceol
			\coqdocindent{1.00em}
			\coqdockw{match} \coqdocvar{l} \coqdockw{with}\coqdoceol
			\coqdocindent{2.00em}
			(\coqdocvar{nil}, \coqdocvar{nil}) \ensuremath{\Rightarrow} \coqdocvar{nil}\coqdoceol
			\coqdocindent{1.00em}
			\ensuremath{|} (\coqdocvar{nil}, \coqdocvar{l2}) \ensuremath{\Rightarrow} \coqdocvar{l2}\coqdoceol
			\coqdocindent{1.00em}
			\ensuremath{|} (\coqdocvar{l1}, \coqdocvar{nil}) \ensuremath{\Rightarrow} \coqdocvar{l1}\coqdoceol
			\coqdocindent{1.00em}
			\ensuremath{|} (\coqdocvar{a} ::\coqdocvar{tl}, \coqdocvar{b}::\coqdocvar{tl'}) \ensuremath{\Rightarrow} \coqdockw{if} \coqdocvar{a} <=? \coqdocvar{b} \coqdockw{then} \coqdocvar{a} :: (\coqdocvar{Merge} (\coqdocvar{tl}, \coqdocvar{b}::\coqdocvar{tl'})) \coqdoceol
			\coqdocindent{12.00em}
			\coqdockw{else} \coqdocvar{b} :: (\coqdocvar{Merge} (\coqdocvar{a} ::\coqdocvar{tl}, \coqdocvar{tl'}))\coqdoceol
			\coqdocindent{1.00em}
			\coqdockw{end}.\coqdoceol
			\coqdocemptyline
			\coqdocnoindent
			\coqdockw{Lemma} \coqdocvar{Merge\_nil}: \coqdockw{\ensuremath{\forall}} \coqdocvar{l}, \coqdocvar{Merge} (\coqdocvar{nil}, \coqdocvar{l}) = \coqdocvar{l}.\coqdoceol
			\coqdocemptyline
			\coqdocnoindent
			\coqdockw{Lemma} \coqdocvar{Merge\_In} : \coqdockw{\ensuremath{\forall}} \coqdocvar{x} (\coqdocvar{l1}:\coqdocvar{list} \coqdocvar{nat}) (\coqdocvar{l2}: \coqdocvar{list} \coqdocvar{nat}),  ((\coqdocvar{In} \coqdocvar{x} \coqdocvar{l1}) \ensuremath{\lor} (\coqdocvar{In} \coqdocvar{x} \coqdocvar{l2})) \ensuremath{\leftrightarrow} \coqdoceol
			\coqdocindent{20.50em}
			\coqdocvar{In} \coqdocvar{x} (\coqdocvar{Merge} (\coqdocvar{l1}, \coqdocvar{l2})).\coqdoceol
			\coqdocemptyline
			\coqdocnoindent
			\coqdockw{Lemma} \coqdocvar{count\_Merge} : \coqdockw{\ensuremath{\forall}} \coqdocvar{l1} \coqdocvar{l2} \coqdocvar{x},  \coqdocvar{count} \coqdocvar{x}  (\coqdocvar{Merge} (\coqdocvar{l1}, \coqdocvar{l2})) = \coqdoceol
			\coqdocindent{19.50em}
			\coqdocvar{count} \coqdocvar{x} \coqdocvar{l1}  + \coqdocvar{count} \coqdocvar{x} \coqdocvar{l2}.\coqdoceol
			\coqdocemptyline
			\coqdocnoindent
			\coqdockw{Lemma} \coqdocvar{permutation\_Merge} : \coqdockw{\ensuremath{\forall}} \coqdocvar{l1} \coqdocvar{l2}, \coqdocvar{permutation} (\coqdocvar{Merge} (\coqdocvar{l1}, \coqdocvar{l2})) (\coqdocvar{l1} ++ \coqdocvar{l2}).\coqdoceol
			\coqdocemptyline
			\coqdocnoindent
			\coqdockw{Lemma} \coqdocvar{sorted\_le} : \coqdockw{\ensuremath{\forall}} \coqdocvar{x} \coqdocvar{l}, \coqdocvar{IsSorted} (\coqdocvar{x} :: \coqdocvar{l}) \ensuremath{\rightarrow} \coqdockw{\ensuremath{\forall}} \coqdocvar{y}, \coqdocvar{In} \coqdocvar{y} \coqdocvar{l} \ensuremath{\rightarrow} \coqdocvar{x} \ensuremath{\le} \coqdocvar{y}.\coqdoceol
			\coqdocemptyline
			\coqdocnoindent
			\coqdockw{Lemma} \coqdocvar{Merge\_sorted} : \coqdockw{\ensuremath{\forall}} (\coqdocvar{l1}:\coqdocvar{list} \coqdocvar{nat}) (\coqdocvar{l2}: \coqdocvar{list} \coqdocvar{nat}), \coqdocvar{IsSorted} \coqdocvar{l1} \ensuremath{\rightarrow} \coqdoceol
			\coqdocindent{18.00em}
			\coqdocvar{IsSorted} \coqdocvar{l2} \ensuremath{\rightarrow} \coqdocvar{IsSorted}(\coqdocvar{Merge} (\coqdocvar{l1},\coqdocvar{l2})).\coqdoceol
			\coqdocemptyline
			\coqdocnoindent
			\coqdockw{Lemma} \coqdocvar{sorted\_cons\_del}: \coqdockw{\ensuremath{\forall}} \coqdocvar{a} \coqdocvar{b} \coqdocvar{l}, \coqdocvar{IsSorted} (\coqdocvar{a} :: \coqdocvar{b} :: \coqdocvar{l}) \ensuremath{\rightarrow} \coqdocvar{IsSorted} (\coqdocvar{a} :: \coqdocvar{l}).\coqdoceol
			\coqdocemptyline
			\coqdocnoindent
			\coqdockw{Lemma} \coqdocvar{sorted\_app\_del}: \coqdockw{\ensuremath{\forall}} \coqdocvar{a} \coqdocvar{l1} \coqdocvar{l2}, \coqdocvar{IsSorted} (\coqdocvar{a} :: (\coqdocvar{l1} ++ \coqdocvar{l2})) \ensuremath{\rightarrow} \coqdocvar{IsSorted} (\coqdocvar{a} :: \coqdocvar{l2}).\coqdoceol
			\coqdocemptyline
			\coqdocnoindent
			\coqdockw{Lemma} \coqdocvar{PatSplit\_sorted} : \coqdockw{\ensuremath{\forall}} \coqdocvar{a} \coqdocvar{b} \coqdocvar{l} \coqdocvar{l1} \coqdocvar{l2}, \coqdocvar{a} \ensuremath{\le} \coqdocvar{b} \ensuremath{\rightarrow} \coqdocvar{IsSorted} (\coqdocvar{a} :: \coqdocvar{l1} ++ [\coqdocvar{b}]) \ensuremath{\rightarrow} \coqdoceol
			\coqdocindent{18.00em}\coqdocvar{IsSorted} (\coqdocvar{fst} (\coqdocvar{PatSplit} \coqdocvar{a} \coqdocvar{b} \coqdocvar{l} \coqdocvar{l1} \coqdocvar{l2})).\coqdoceol
			\coqdocemptyline
			\coqdocnoindent
			\coqdockw{Lemma} \coqdocvar{PatSplit\_A\_notempty} : \coqdockw{\ensuremath{\forall}} \coqdocvar{a} \coqdocvar{b} \coqdocvar{l} \coqdocvar{Y} \coqdocvar{Z}, \coqdocvar{fst}(\coqdocvar{PatSplit} \coqdocvar{a} \coqdocvar{b} \coqdocvar{l} \coqdocvar{Y} \coqdocvar{Z}) \ensuremath{\not=} \coqdocvar{nil}.\coqdoceol
			\coqdocemptyline
			\coqdocnoindent
			\coqdockw{Lemma} \coqdocvar{PatSplit\_permutation} : \coqdockw{\ensuremath{\forall}} \coqdocvar{a} \coqdocvar{b} \coqdocvar{l} \coqdocvar{Y} \coqdocvar{Z}, \coqdocvar{permutation} (((\coqdocvar{a} :: \coqdocvar{b} :: \coqdocvar{l}) ++ \coqdocvar{Y}) ++ \coqdocvar{Z}) ((\coqdocvar{fst} \coqdoceol
			\coqdocindent{8.00em}(\coqdocvar{PatSplit} \coqdocvar{a} \coqdocvar{b} \coqdocvar{l} \coqdocvar{Y} \coqdocvar{Z})) ++ (\coqdocvar{snd} (\coqdocvar{PatSplit} \coqdocvar{a} \coqdocvar{b} \coqdocvar{l} \coqdocvar{Y} \coqdocvar{Z}))).\coqdoceol
			\coqdocemptyline
			\coqdocnoindent
			\coqdockw{Theorem} \coqdocvar{PatSort3\_is\_correct} : \coqdocvar{is\_a\_sorting\_algorithm} \coqdocvar{PatSort3}.\coqdoceol
			
		\end{document}